\documentclass[a4paper]{llncs}
\usepackage{graphicx} 
\usepackage{float} 
\usepackage[utf8]{inputenc}
\usepackage{amsmath}
\usepackage{amssymb}

\usepackage{amsthm}
\usepackage{times}
\usepackage{url}
\usepackage{xspace}
\usepackage{comment}
\usepackage{color}
\usepackage{thmtools, thm-restate}
\usepackage{algorithm}
\usepackage[noend]{algorithmic}
\usepackage{booktabs}
\usepackage{tabulary}
\usepackage[margin=1in]{geometry}
\usepackage{multirow}
\usepackage{breakcites} 
\usepackage{natbib}
\declaretheorem[name=Theorem]{thm}

\declaretheorem[name=Lemma]{lem}

\title{On the representation of de Bruijn graphs$^*$}
\author{Rayan Chikhi$^{1,6}$, Antoine Limasset$^2$, Shaun Jackman$^3$, Jared T. Simpson$^4$ and Paul Medvedev$^{1,5,6,\dagger}$}
\institute{
$^1$Department of Computer Science and Engineering, The Pennsylvania State University, USA\\
$^2$ENS Cachan Brittany, Bruz, France\\ 
$^3$Canada's Michael Smith Genome Sciences Centre, Canada\\
$^4$Ontario Institute for Cancer Research, Toronto, Canada\\
$^5$Department of Biochemistry and Molecular Biology, The Pennsylvania State University, USA\\
$^6$Genome Sciences Institute of the Huck, The Pennsylvania State University, USA \\
}

\newcommand{\setof}[2]{\{#1\;:\;#2\}}

\newcommand{\extr}{\overleftarrow{ext}}
\newcommand{\extf}{\overrightarrow{ext}}
\newcommand{\f}{f} 

\newcommand{\idxset}{\text{\sc{idx}}}

\newcommand{\kmer}{$k$-mer\xspace}
\newcommand{\lmer}{$\ell$-mer\xspace}
\newcommand{\kmers}{$k$-mers\xspace}
\newcommand{\lmers}{$\ell$-mers\xspace}
\newcommand{\neighb}{\text{\sc nbr}\xspace}
\newcommand{\const} {\text{\sc const}\xspace}

\newcommand{\mem} {\text{\sc memb}\xspace}
\newcommand{\dbg}   {\text{dBG}\xspace}
\newcommand{\dbgs}   {\text{dBGs}\xspace}
\newcommand{\nds}   {\text{NDS}\xspace}
\newcommand{\sai} {\text{SAI}\xspace}
\newcommand{\sais} {\text{SAIs}\xspace}

\newcommand{\defn}[1]{\emph{#1}}

\newcommand{\lmin}{\text{Lmin}\xspace}
\newcommand{\rmin}{\text{Rmin}\xspace}

\definecolor{orange}{rgb}{1,0.5,0}

\newcommand{\DSK}{\text{\textsc{DSK}}\xspace}
\newcommand{\dbgfm}{\text{\textsc{dbgfm}}\xspace}
\newcommand{\bcalm}{\text{\textsc{bcalm}}\xspace}

\newcommand{\jsnote}[1]{\textcolor{green}{JS -- #1}}
\newcommand{\pmnote}[1]{\textcolor{blue}{PM -- #1}}
\newcommand{\rcnote}[1]{\footnote{\textcolor{cyan}{RC -- #1}}}

\newcommand{\revision}[1]{#1}
\newcommand{\journal}[1]{#1}

\newcommand{\arxiv}[1]{#1} 

\mathcode`l="8000
\begingroup
\makeatletter
\lccode`\~=`\l
\DeclareMathSymbol{\lsb@l}{\mathalpha}{letters}{`l}
\lowercase{\gdef~{\ifnum\the\mathgroup=\m@ne \ell \else \lsb@l \fi}}%
\endgroup

\begin{document}

\maketitle
\begin{abstract}
The de Bruijn graph plays an important role in bioinformatics, especially in the context of \emph{de novo} assembly. 
However, the representation of the de Bruijn graph in memory is a computational bottleneck for many assemblers.
Recent papers proposed a navigational data structure approach in order to improve memory usage.
We prove several theoretical space lower bounds to show the limitation\revision{s} of these types of approaches.
We further design and implement a general data structure (\dbgfm) and demonstrate its use on a human whole-genome dataset,
achieving space usage of 1.5 GB and a 46\% improvement over previous approaches.
As part of \dbgfm, we develop the notion of frequency-based minimizers and show how it can be used to 
enumerate all maximal simple paths of the de Bruijn graph using only 43 MB of memory.
Finally, we demonstrate that our approach can be integrated into an existing assembler by modifying the
ABySS software to use \dbgfm.
\let\thefootnote\relax\footnotetext{$^*$A preliminary version of some of these results appeared in \citet{Chikhi14}.}
\let\thefootnote\relax\footnotetext{$^\dagger$Corresponding author, pashadag@cse.psu.edu}
\end{abstract}

\section{Introduction}

\emph{De novo} assembly continues to be one of the fundamental problems in bioinformatics,
with new datasets coming from projects such as the Genome10K~\citep{genome10k}.
The task is to reconstruct an unknown genome sequence from a set of short sequenced fragments.
Most state-of-the-art assemblers 
(e.g.~\citet{allpathslg,soapdenovo,spades,velvet})
start by building a de Bruijn graph (\dbg)~\citep{Pevzner1989,idury_new_1995}, 
which is  a directed graph where each node is a distinct $k$-mer present in the input fragments, and an edge is present between two $k$-mers when they share an exact $(k-1)$-overlap.
The de Bruijn graph is the basis of many steps in assembly, including path compression, bulge removal, graph simplification, and repeat resolution~\citep{miller_assembly_2010}. 
In the workflow of most assemblers, the  graph must, at least initially, reside in memory;
thus, for large genomes, memory is a computational bottleneck.
For example, the graph of a human genome consists of nearly three billions nodes and edges and assemblers require computers with hundreds of gigabytes of memory~\citep{allpathslg,soapdenovo}. 
Even these large resources can be insufficient for many genomes,
such as the 20 Gbp white spruce.
Recent assembly  required a distributed-memory approach 
and  around a hundred large-memory servers, collectively storing a $4.3$ TB de Bruijn graph data structure~\citep{spruce}.

Several articles have pursued the question of whether smaller data structures could be designed to make large genome assembly more accessible~\citep{conway_succinct_2011,ye12,titus,minia,sadakane}.
\citet{conway_succinct_2011} gave a lower bound on the number of bits required to encode a de Bruijn graph consisting of $n$ $k$-mers:
$\Omega(n\lg n)$  (assuming $4^k > n$). 
However,  two groups  independently observed 
that assemblers use \dbgs in a very narrow manner~\citep{minia,sadakane}
and proposed 
a data structure that is able to return the set of neighbors of a given node but is 
not necessarily able to determine if that node is in the graph.
We refer to these as {\em navigational data structures} (\nds).
The navigational data structures proposed in~\citet{minia,sadakane} require $O(n\lg k)$ and $O(n)$\footnote{The paper only showed the number of bits is $O(n\lg n)$. However, the authors  recently indicated in a blog post~\citep{boweblogpost} how the dependence on $\lg(n)$ could be removed, though the result has not yet been published.}  
 bits (respectively),
beating the Conway-Bromage lower bound both in theory and in practice~\citep{minia}. 

What is the potential of these types of approaches to further reduce memory usage?
To answer this question, we first formalize the notion of a navigational data structure and 
then show that any \nds requires at least $3.24n$ bits.
This result leaves a gap with the known upper bounds;
however, even if a \nds could be developed to meet this bound,
could we hope to do better on inputs that occur in practice?
To answer this, we consider a very simple class of inputs: simple paths.
We show that on these inputs (called linear \dbgs), 
there are both navigational and general data structures that 
asymptotically use $2n$ bits and give matching lower bounds.
While \dbgs occurring in practice are not linear, they can nevertheless be often decomposed into a small collection of long simple paths (where all the internal nodes have in- and out-degree of 1).
Could we then take advantage of such a decomposition to develop a data structure that can achieve close to $2n$ bits on practical inputs?

We describe and implement a data structure (\dbgfm) to represent de Bruijn graphs in low memory.
The first step of the construction uses 
existing $k$-mer counting software to transform, in constant memory, the input sequencing dataset to a list of 
\kmers (i.e. nodes) stored on disk~\citep{dsk}.
The second step is a novel low memory algorithm that enumerates all the maximal
simple paths without loading the whole graph in memory. 
We achieve this through the use of 
non-lexicographic minimizers,
ordered based on their frequency in the data.
Finally, we use the FM-index~\citep{fmindex} to store the simple paths in memory 
and answer membership and neighborhood queries.

We prove that as the number of simple paths decreases,
the space utilization of \dbgfm  approaches $2n$ bits.
In practice, \dbgfm uses $4.76n$ bits on a human whole-genome dataset and $3.53n$ bits on a human chr14 dataset, 
improving the state-of-the-art~\citep{cascading}
by 46\% and 60\%, respectively.
We demonstrate the efficiency of frequency-based minimizers by 
collapsing the \dbg of the human whole-genome dataset using only $43$ MB of memory.
Finally, we show how \dbgfm can be integrated into an existing assembler by modifying the ABySS software~\citep{abyss}
to use \dbgfm instead of a hash table.

\section{Previous Work}

In the last three years, 
several papers and assemblers have explored novel data structures designed to reduce the space usage of \dbgs, 
and we provide a brief summary of the results here.


ABySS was one of the first genome assemblers capable of representing large \dbgs~\citep{abyss}. 
It uses an open-addressing hash table that stores the $k$-mers of the graph in the keys. 
The edges can be inferred from the nodes and do not need to be stored.
For every \kmer, ABySS uses $2k$ bits to store the \kmer, plus an additional 43 bits of associated data
(stored in 64 bits for ease of implementation).
Therefore, in total, the space usage of the \dbg data structure in ABySS is $(l^{-1}(2k+64))$ bits per $k$-mer, where $l$ is the load factor of the hash table (set to 0.8).
In the following, we focus on the space needed to store just the \dbg,
since the type of associated data varies greatly between different assemblers.

\citet{conway_succinct_2011}
gave a $\lg{4^k \choose n}$ bits lower bound for representing a \dbg 
and demonstrated a sparse bit array data structure that comes close to achieving it.
They used an edge-centric definition of the \dbg 
(where edges are all the $(k+1)$-mers, and nodes are prefixes and suffixes of length $k$), 
but their results trivially translate to node-centric \dbgs 
by storing $k$-mers instead of $(k+1)$-mers.
For a dataset with $k=27$ and $12\cdot 10^9$ edges (i.e. ($k+1$)-mers), 
their theoretical minimum space is 22 bits per edge 
while their implementation achieves 28.5 bits per edge.

Later work explored the trade-offs between the amount of information retained 
from the de Bruijn graph and the space usage of the data structure.
\citet{ye12}
showed that a graph equivalent to the de Bruijn graph can be stored in a hash table by sub-sampling $k$-mers. 
The values of the hash table record sequences that would correspond to paths between $k$-mers in the de Bruijn graph. 
The theoretical memory usage of this approach is $\Omega(k/g)$ bits per $k$-mer, where $g$ is the distance between consecutive sampled $k$-mers.
\citet{titus} proposed a practical lossy approximation of the de Bruijn graph that stores the nodes in a Bloom filter~\citep{bloomfilters}. 
They found that a space usage of $4$ bits per $k$-mer provided a reasonable approximation of the de Bruijn graph for their purpose (partitioning and down-sampling DNA sequence datasets). 
Yet, the structure has not yet been directly applied to \emph{de novo} assembly.

\citet{minia} 
built upon the structure of~\citet{titus} by 
additionally storing the set of Bloom filter false positives (false neighbors of true nodes in the graph).
In this way, their structure is no longer lossy.
They obtained a navigational data structure that allowed the assembler to exactly
enumerate the in- and out-neighbors of any graph node in constant time. However, the structure does not support node membership queries, and also does not support storing associated data to $k$-mers. 
The theoretical space usage is $(1.44\lg(\frac{16k}{2.08})+2.08)$ bits per $k$-mer, under certain 
assumptions about the false positive rate of the Bloom filter.
%
This corresponds to $13.2$ bits per $k$-mer for $k=27$. 

The structure has recently been improved by~\citet{cascading} with cascading Bloom filters, 
replacing the hash table by a cascade of Bloom filters. 
In theory, if an infinite number of Bloom filters is used, this scheme would require $7.93$ bits per $k$-mer independently of $k$. 
The authors show that using only $4$ Bloom filters is satisfactory in practice, yet they do not provide a formula for the theoretical space usage in this case. For $k=27$ and $2.7\cdot 10^9$ nodes, they computed that their structure uses $8.4$ bits per $k$-mer.
\citet{sadakane} used a tree variant of the Burrows-Wheeler transform~\citep{bwt} to support identical operations. 
They describe a theoretical navigational data structure for representing the \dbg of a set of input sequences that uses 
a space $4m + M \lg(m) + o(m)$ bits, where $M$ is the number of input strings and $m$ the number of graph edges. 
Note that the space is independent of $k$. \revision{Another data structure based on a similar principle has been recently proposed~\citep{kfmindex}}.


\journal{
In addition to studying the representation of de Bruijn graphs, several articles have designed efficient algorithms for constructing the graph.
\citet{yanglipaper} proposed an algorithm based on minimizers that, given a set of reads, outputs to the disk both the edges and the nodes of the de Bruijn graph (essentially performing a $k$-mer counting step). 
\citet{movahedi2012novo} also used minimizers to reduce the memory usage of de Bruijn graph compaction. Their approach consists in partitioning the initial graph into disjoint components, assigning each $k$-mer to a \emph{slice} via its minimizer. Slices are then compacted in no specific order, and the resulting compacted slices are merged together into a final compacted graph. The improvement in terms of memory reduction achieved by this two-stage compaction approach was not analyzed in~\citet{movahedi2012novo}.
Finally, several methods have been recently proposed to construct a compacted de Bruijn graph in linear time and memory from a suffix array~\citep{cazaux,sibelia} or a suffix tree~\citep{splitmem}.
}

\section{Preliminaries}
We assume, for the purposes of this paper, that all strings are over the alphabet $\Sigma = \{A,C,G,T\}$.
A string of length $k$ is called a $k$-mer and $U$ is the universe of all $k$-mers, i.e. $U=\Sigma^k$.
The binary relation $u \rightarrow v$ between two strings 
denotes an exact suffix-prefix overlap of length $(k-1)$ between $u$ and $v$.
For a set of $k$-mers $S$, the \defn{de Bruijn graph} of $S$ is a directed graph such that the nodes are exactly the $k$-mers in $S$
and the edges are given by the  $\rightarrow$ relation. 
We define $S$ to be a \emph{linear} \dbg if there exists a string $x$ where 
all the $(k-1)$-mers of $x$ are distinct and 
$S$ is the set of \kmers present in $x$.
Equivalently, $S$ is a linear \dbg if and only if the graph is a simple path.
The de Bruijn graph of a string $s$ is the de Bruijn graph of all the $k$-mers in $s$.
We adopt the node-centric definition of the de Bruijn graph, where the edges are implicit given the vertices; 
therefore, we use the terms de Bruijn graph and a set of $k$-mers interchangeably.

For a node $x$ in the de Bruijn graph, let $\extr(x)$ be its \revision{four potential} in-neighbors 
(i.e. $\extr(x) = \setof{y}{y \in \Sigma^k, y\rightarrow x}$
) 
and $\extf(x)$ be its \revision{four potential} out-neighbors. 
Let $ext(x) = \extf(x) \cup \extr(x)$. 
For a given set of $k$-mers $S$, let $ext(S) = \{ ext(x), x \in S\}$ (similarly for $\extf(S)$ and $\extr(S)$).

We will need some notation for working with index sets, 
which is just a set of integers that is used to select a subset of elements from another set.
Define $\idxset(i,j)$ as a set of all index sets that select $j$ out of $i$ elements.
Given a set of $i$ elements $Y$ and $X \in \idxset(i,j)$, 
we then write $Y[X]$ to represent the subset of $j$ elements out of $Y$, as specified by $X$.
We assume that there is a natural ordering on the elements of the set $Y$, e.g.
if $Y$ is a set of strings, then the ordering might be the lexicographical one.

The families of graphs we will use to construct the lower bounds of 
Theorems~\ref{thm:generalgraphslb} and \ref{thm:lineargraphslb} 
have $k$ be a 
polylogarithmic function of $|S|$, i.e. $k = O(\log^c |S|)$ for some $c$.
We note that in some cases, higher lower bounds could be obtained using families of graphs with $k = \Theta(|S|)$;
however, we feel that such values of $k$ are unrealistic given the sequencing technologies.
On one hand, the value of $k$ is a bounded from above by the read length, 
which is experimentally independent of the number of \kmers.
On the other hand, $k$ must be at least $\log_4(|S|)$ in order for there to be at least $|S|$ distinct $k$-mers.

\section{Navigational data structures}

We use the term {\em membership data structure} to refer to a way of representing a \dbg 
and answering \kmer membership queries.
We can view this as a pair of algorithms: $(\const,\mem)$.
The \const algorithm takes a set of \kmers $S$ (i.e. a \dbg) and outputs a bit string.
We call \const a constructor, since it constructs a representation of a \dbg.
The \mem algorithm takes as input a bit string and a \kmer $x$ and outputs true or false.
Intuitively, \mem takes a representation of a \dbg created by \const and outputs whether a 
given \kmer is present.
Formally, we require that for all $x \in \Sigma^k$, $\mem(\const(S),x)$ is true if and only if $x\in S$.
An example membership data structure, as used in ABySS, is one where the \kmers are 
put into a hash table (the \const algorithm) and
membership queries are answered by hashing the \kmer to its location in the table (the \mem algorithm).

Recently, it was observed that most assemblers use the \mem algorithm in a limited way~\citep{minia,sadakane}. 
They do not typically ask for membership of a vertex that is not in $ext(S)$, but, 
instead, ask for the neighborhood of nodes that it already knows are in the graph. 
We formalize this idea by introducing the term navigational data structure (\nds), 
inspired by the similar idea of performing navigational queries on trees~\citep{xbwt}.
An \nds is a pair of algorithms, \const and \neighb.
As before, 
\const takes a set of \kmers and outputs a bit string.
\neighb takes a bit string and a \kmer, and outputs a set of \kmers.
The algorithms  must satisfy that for every \dbg $S$ and a \kmer $x \in S$,
$\neighb(\const(S),x)=ext(x) \cap S$.
Note that if $x \notin S$, then the behavior of $\neighb(\const(S),x)$ is undefined.
We observe that a membership data structure immediately implies a \nds because a \neighb query 
can be reduced to eight \mem queries.

To illustrate how such a data structure can be useful, consider a program that can enumerate nodes using external memory (e.g. a hard drive or a network connection). 
Using external memory to navigate the graph by testing node membership would be highly inefficient because of long random access times. 
However, it is acceptable to get a starting node from the device and access the other nodes using the proposed data structure.

There are several important aspects of both a navigational and membership data structures,
including the space needed to represent the output of the constructor, 
the memory usage and running time of the constructor,
and the time needed to answer either neighborhood or membership queries.
For proving space lower bounds, we make no restriction on the other resources so that our bounds hold more generally.
However, adding other constraints 
(e.g. query time of $\lg n$) may allow us to prove higher lower bounds and
is an interesting area for future work.

\section{Navigational data structure lower bound for de Bruijn graphs \label{sec:lbgen}}

In this section, we prove that a navigational data structure on de Bruijn graphs needs at least 3.24 bits per $k$-mer to represent the graph:
\begin{restatable}{thm}{generalgraphslb}
    \label{thm:generalgraphslb}
	Consider an arbitrary \nds and let $\const$ be its constructor.
    For any $0 < \epsilon < 1$,
	there exists a $k$ and $x\subseteq \Sigma^k$ such that $|\const(x)| \geq |x|\cdot (c - \epsilon) $,
	where 
	$c = 8 - 3\lg3 \approx 3.25$.
\end{restatable}
We will first present a high level overview of the proof strategy followed by the formal proof afterwards.

\subsection{Proof strategy overview}
\revision{Our proof strategy is to construct a family of graphs, for which the number of navigational data structures is at least the size of the family.} The full proof of the theorem is in the Section~\ref{sec:lbgen:details}, however,
we will describe the construction used and the overall outline here.
Our first step is to construct a large \dbg $T$ and later we will choose subsets as members of our family.
Fix an even $k \ge 2$, 
let $\ell = k/2 $, and let $m = 4^{\ell - 1}$.
$T$ will be defined as the union of $\ell + 1$ levels, $T = \cup_{0 \le i \le \ell} T_i$.
For $0 \leq i \leq \ell$, we define the $i^{\text{th}}$ level as 
$T_i = \setof{\text{``A$^{\ell-i}$T}\alpha"}{\alpha \in \Sigma^{ i + \ell - 1 }}$.
Observe that $T_i = \extf(T_{i-1})$, for $1 \leq i \leq \ell$.
See Figure~\ref{fig:lb} for a small example.

\arxiv{

\begin{figure}[t]
\centering
\includegraphics[scale=0.8]{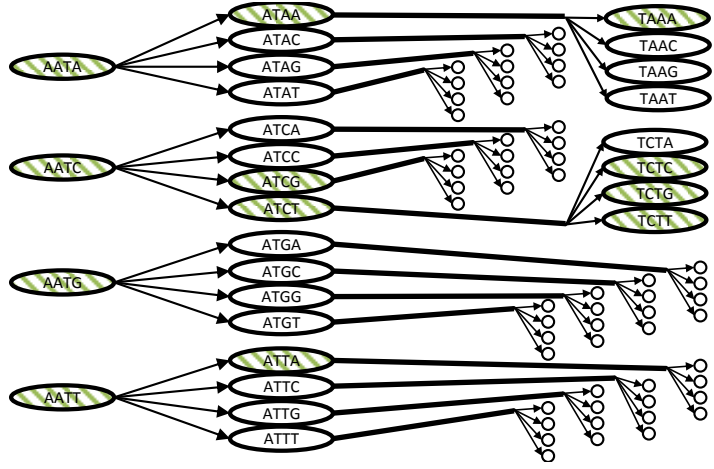}

\caption{
Example of lower bound construction for $k=4$. 
The figure shows $T$ along with some of the node labels.
The four nodes on the left form $T_0$, the 16 nodes in the middle are $T_1$,
and the nodes on the right are $T_2$.
For space purposes, some of the edges from $T_1$ to $T_2$ are grouped together.
An example of a member from the family is shown with shaded vertices.
Note that there are four vertices at each level, and together
they form a subforest of $T$.
\label{fig:lb}}
\end{figure}
}


We focus on constructing \dbgs that are subsets of $T$ because $T$ has some desirable properties.
In fact, one can show that the set of $k$-mers $T$ induces a forest in the \dbg of $\Sigma^k$ 
(Lemmas~\ref{lem:tree2} and \ref{lem:tree3} in the Appendix).
Each member of our family will be a subforest of $T$ that contains $m$ vertices from every level.

Formally, suppose we are given a sequence of index sets $X = X_1, \ldots, X_\ell$ 
where every index set is a member of $\idxset(4m,m)$.
Each index set will define the subset of vertices we select from a level, and 
we define 
$L^X_0 = T_0$ and $L^X_i = \extf(L^X_{i-1})[X_i]$, for $1 \leq i \leq \ell$.
Note that $L^X_i \subseteq T_i$. 
In this manner, the index sets define a set of $k$-mers $s(X) = \cup_{0\leq i \leq l} L^X_i$.
Finally, the family of graphs which we will use for our proof is given by:
$$
S^k = \setof{s(X_1,\ldots, X_\ell)}{ \ell = k/2, m = 4^{\ell-1}, X_i \in \idxset(4m,m)}
$$

To prove Theorem~\ref{thm:generalgraphslb}, 
we first show that each of the \dbgs of our family have the same amount of \kmers:
\begin{restatable}{prop}{sizen}
    \label{lem:sizen}
    For all $x \in S^k$, $|x| = 4^{\ell - 1}(\ell + 1)$.
\end{restatable}
Next, we show that each choice of $X$ leads to a unique graph $s(X)$ (Lemma~\ref{lem:unique})
and use it to show that the numbers of graphs in our family is large,
relative to the number of \kmers in each set:
\begin{restatable}{prop}{size}
    \label{lem:size}
	$|S^k| = { 4m \choose m } ^{\ell} \geq 2^{(c-\epsilon_0)\ell m}$, where $c= 8 - 3\lg 3$ and 
	$\epsilon_0 = \lg(12m) / m$.
\end{restatable}
Finally, we need to show that for any two graphs in the family,
there is at least one $k$-mer that appears in both graphs but with different neighbors:
\begin{restatable}{prop}{goodelement}
	\label{lem:goodelement}
	Let $x = s(X) \in S^k$ and 
	    $y = s(Y) \in S^k$
	be two distinct elements in $S^k$.
	Then, there exists a $k$-mer  $u \in x \cap y$ 
	such that $\extf(u) \cap x \ne \extf(u) \cap y$.
\end{restatable}

The proof of Theorem~\ref{thm:generalgraphslb} now follows 
by using the pigeonhole principle to argue that the number of navigational data structures
must be at least the size of our family,
giving a lower bound on the bits per $k$-mer.

\subsection{Proof details}\label{sec:lbgen:details}
We now give a formal proof of the three Properties and the Theorem.
\begin{restatable}{lem}{treeii}
    \label{lem:tree2}
    Let $y\in T$.
    There exists a unique $0 \le i \le \ell$ such that $y\in T_i$.
\end{restatable}
\begin{proof}
    Take two arbitrary levels $i_1 < i_2$ and two arbitrary vertices in those levels, 
    $x_1 \in T_{i_1}$ and 
	$x_2 \in T_{i_2}$.
	Let $z \in \{1,2\}$.
	The $k$-mer represented by $x_z$ is 
	$\text{``A$^{\ell - i_z + 1}$T$\alpha_z$"}$,
	where $\alpha_z$ is some string.
	At position $\ell - i_1 + 1$, $x_1$ has a T, while $x_2$ has an A.
	Therefore, $x_1 \ne x_2$ and the lemma follows.
\end{proof}

\begin{restatable}{lem}{treeiii}
	\label{lem:tree3}
	For all distinct $x_1$ and $x_2$ in $T$ that are not in the last level ($T_\ell$),
	$\extf(x_1) \cap \extf(x_2) = \emptyset$.
\end{restatable}
\begin{proof}
	By Lemma~\ref{lem:tree2},
	there exist unique levels $i_1$ and $i_2$ such that $x_1 \in T_{i_1}$ and $x_2 \in T_{i_2}$.
	We first observe that $\extf(x_z) \in T_{i_z+1}$, for $z \in \{1,2\}$.
	If it is the case that $i_1 \ne i_2$, then Lemma~\ref{lem:tree2} 
	applied to the vertices in the extensions prove the lemma.
	Now suppose that $i_1 = i_2$, and we write $i = i_1$.
	Then, for $z \in \{1,2\}$, 
	the $k$-mer represented by $x_z$ is 
	$\text{``A$^{\ell - i }$T$\alpha_z$"}$,
	where $\alpha_z$ is a ($\ell + i - 1 $)-mer and $\alpha_1 \ne \alpha_2$.
	We can then write the extensions as 
$\extf(x_z) = 
\setof
{
\text{``A$^{\ell - i-1}$T$\alpha_z\beta$"}
} {
\beta \in \{A,C,G,T\}
}$.
Because $\alpha_1 \ne \alpha_2$, the sets $\extf(x_1)$ and $\extf(x_2)$ share no common elements.
\end{proof}

Property~\ref{lem:sizen} now follows directly from Lemmas~\ref{lem:tree2} and \ref{lem:tree3}.
To prove Property~\ref{lem:size}, we need the following two lemmas:
\begin{restatable}{lem}{unique}
    \label{lem:unique}
	Let $X = X_1,\ldots, X_\ell$ and
	$Y = Y_1,\ldots, Y_\ell$ be two sequences of index sets.
	Then $s(X) = s(Y)$ if and only if $X = Y$.
\end{restatable}
\begin{proof}
    Since the construction is fully deterministic and depends only on the index sets, 
	then $X = Y$ immediately implies $s(X) = s(Y)$.
	For the other direction,
	suppose that $X \ne Y$.
	Let $i>0$ be the smallest index such that $X_i \ne Y_i$.
	Then there exists a vertex $y$ such that $y \in L^X_i$ but $y \notin L^Y_i$.
	Since $y$ is in  $T_i$ but not in $L^Y_i$, Lemma~\ref{lem:tree2} implies that $y\notin s(Y)$.
\end{proof}

\begin{restatable}{lem}{choosebound}
	\label{lem:choosebound}
	For all $m>0$, 
	${ 4m \choose m} \ge 2^{(c-\epsilon_0)m}$, where
    $c = 8 - 3\lg3$ and
	$\epsilon_0 = \lg(12m) / m$.
\end{restatable}
\begin{proof}
	Follows directly from an inequality of~\citet{sondow_problem}:
    $ { rm \choose m} > \frac{2^{cm}}{4m(r-1)} $.
\end{proof}

\size*
\begin{proof}
	Lemma~\ref{lem:unique}
	tells us that the size of $S^k$ is the number of possible ways one could choose $X_1,\ldots,X_\ell$
	during the construction of each element $s(X_1,\ldots,X_\ell)$.
	The choice for each $X_i$ is independent, and there are ${4m \choose m}$ possibilities.
	Hence, there are ${4m \choose m}^\ell$ total choices.
	The inequality follows from Lemma~\ref{lem:choosebound}.
\end{proof}

We can now prove Property~\ref{lem:goodelement} and Theorem~\ref{thm:generalgraphslb}:
\goodelement*
\begin{proof}
	By Lemma~\ref{lem:unique},  $X \ne Y$.
	Let $i$ be the smallest index such that $X_i \ne Y_i$, 
	and let $v$ be an element in $L^X_i$ but not in $L^Y_i$.
	By construction, 
	there exists a vertex $u \in L^X_{i-1}$ (and hence in $L^Y_{i-1}$) such that $v \in \extf(u)$.
	Lemma~\ref{lem:tree2} tells us that $v$ is not in $y$ and hence
	$u$ satisfies the condition of the lemma.
\end{proof}

\generalgraphslb*
\begin{proof}
	Assume for the sake of contradiction that for all $x$, $|\const(x)| < |x|(c - \epsilon)$.
	Let 
    $k$ be a large enough integer such that $k > 2c\epsilon^{-1}$ and $\epsilon_0 < (\epsilon(l+1)-c)/l$ holds (with $m,l,\epsilon_0$ as defined above). The second inequality is verified for any large value of $k$, since $\epsilon_0 = \Theta(l/4^l)$ converges to $0$ and $(\epsilon(l+1)-c)/l$ converges to $\epsilon$.
    Let $n = 4^{k/2 - 1}(k/2 + 1)$.
	Consider the outputs of \const on the elements of $S^k$.
	When the input is constrained to be of size $n$, 
	the output must use less than $(c - \epsilon) n$ bits (by Lemma~\ref{lem:sizen}).
	Hence the range of \const over the domain $S^k$ has size less than $2^{(c - \epsilon)n}$.
	At the same time, Lemma~\ref{lem:size} states that there are at least $2^{(c-\epsilon_o)lm}$ elements in $S^k$.
 
    From the inequality $\epsilon_0 < (\epsilon(l+1)-c)/l$ we derive that $(c-\epsilon_0)l > (c-\epsilon)(l+1)$ and thus  $2^{(c-\epsilon_o)\ell m} > 2^{(c-\epsilon)n}$.
    Therefore,
	there must exist distinct $s_1, s_2 \in S^k$ such that $\const(s_1) = \const(s_2)$.
	We can now apply Lemma~\ref{lem:goodelement} to obtain an element $y \in s_1 \cap s_2$
	that is a valid input to $\const(s_1)$ and to $\const(s_2)$.
	Since the two functions are the same, the return value must also the same. 
	However, we know that the out-neighborhoods of $y$ are different in $s_1$ and in $s_2$, 
	hence, one of the results of \neighb on $y$ must be incorrect.
	This contradicts the correctness of $\const$.
\end{proof}

\section{Linear de Bruijn graphs \label{sec:lblin}}
In this section, we study data structures to represent linear de Bruijn graphs.
Though a linear \dbg will never occur in practice,
it is an idealized scenario which lets
us capture how well a data structure can do in the best case.
The bounds obtained here also serve as motivation for our approach in later sections,
where we build a membership data structure whose space usage approaches 
our lower bound from this section the ``closer'' the graph is to being linear.

We can design a naive membership data structure for linear \dbgs.
A linear \dbg with $n$ \kmers corresponds to a string of length $n + k - 1$.
The constructor builds and stores the string from the \kmers,
while the membership query simply does a linear scan through the string.
The space usage  is $2(n+k-1)$ bits. 
The query time is  prohibitively slow, and 
we show in Section~\ref{sec:dbgfm} how to achieve a much faster solution
at the cost of a space increase.

We now prove that a \nds on linear de Bruijn graphs needs
at least $2n$ bits to represent the graph, meaning one cannot do much better
than the naive representation above.
We will first give a proof strategy overview, followed by the formal proof details.

\subsection{Proof strategy overview}

In general, representing all strings of length $n + k - 1$ would take $2 ( n + k - 1)$ bits,
however, not all strings of this length correspond to linear \dbgs.
Fortunately, we can adapt a probabilistic result of~\citet{gagie2012} to
quantify the number of strings of this length that have no duplicate \kmers 
(Lemma~\ref{lem:nb_strings_repeated_kmer}).
Our strategy is similar to that of Section~\ref{sec:lbgen}.
We construct a large family of linear \dbgs such that for any pair of members, there is always a \kmer
that belongs to both but whose neighborhoods are different.
We build the family by taking the set of all strings without duplicate $(k-1)$-mers and 
identifying a large subset having the same starting \kmer.
We 
then 
show that
by increasing the length of the strings and $k$, we can create a family of size arbitrarily close to $4^n$
(Lemma~\ref{lem:nb_strings_non_repeated_kmer}).
Finally, we show that because each string in the family starts with the same \kmer,
there always exists a distinguishing \kmer for any pair of strings.
Using the pigeonhole principle, this implies that
number of navigational data structures must be at least the number of distinct strings: 
\begin{restatable}{thm}{lineargraphslb}
\label{thm:lineargraphslb}
Consider an arbitrary \nds for linear de Bruijn graphs and let \const be its constructor.
Then, for any $0 < \epsilon < 1$, there exists $(n,k)$ and a set of $k$-mers $S$ of cardinality $n$, such that $|\const(S)| \geq 2n (1 - \epsilon) $.
\end{restatable}

Note that our naive membership data structure of $2(n+k-1)$ bits 
immediately implies a \nds of the same size.
Similarly, Theorem~\ref{thm:lineargraphslb}'s  lower bound of $2n$ bits on a \nds 
immediately implies the same bound for membership data structures.
In practice, $k$ is orders of magnitude less than $n$, 
and we view these results as saying that the space usage of membership and navigational 
data structures on linear \dbgs is essentially $2n$ and cannot be improved.

These results, together with Theorem~\ref{thm:generalgraphslb}, 
suggested that the potential of navigational data structures
may be dampened when the \dbg is linear-like in structure.
Intuitively, the advantage of a linear \dbg is that all the \kmers of a path collapse together onto one string
and require storing only one nucleotide per \kmer, except for the overhead of the first \kmer.
If the \dbg is not linear but can still be decomposed into a few paths, 
then we could still take advantage of each path while paying an overhead of only a single \kmer per path.
This in fact forms the basis of our algorithm in Section~\ref{sec:dbgfm}.

\subsection{Proof details}

\begin{restatable}{lem}{nbstringsrepeatedkmer}
\label{lem:nb_strings_repeated_kmer}
The number of DNA strings of length $m$ where each $k$-mer is seen only once is at least $4^m(1-{ m \choose 2 } 4^{-k})$.
\end{restatable}
\begin{proof}
This Lemma was expressed in a probabilistic setting in~\citet{gagie2012}, but we provide a deterministic proof here.
We define a set of strings $S$ and show that it contains all strings with at least one repeated $k$-mer. Let $\overline{s}^k$ be the string obtained by repeating the pattern $s$ as many times as needed to obtain a string of length exactly $k$, possibly truncating the last occurrence. Let 
\begin{align*}
S = \{ s \in \Sigma^m \,| &\exists~(i,j), \journal{1} \leq i < j \leq \journal{(m - k + 1)}, \\
&\exists~t,  |t|=(m-k), \\
&s[1\ldots \journal{j-1}] = t[1\ldots \journal{j-1}], \\
&s[\journal{j}\ldots j+k-1]=\overline{s[\journal{i}\ldots \journal{j-1}]}^k,  \\
&s[\journal{j+k}\ldots m]=t[\journal{j}\ldots m-k]\text{\,if\, \journal{$j < m-k+1$}}\}, 
\end{align*}
as illustrated in Figure~\ref{fig:lblin}.
Let $s'$ be a string of length $m$ which contains at least one repeated $k$-mer. Without loss of generality, there exists two starting positions $i<j$ of identical $k$-mers (\journal{$s'[j\ldots j+k-1]=s'[i\ldots i+k-1]$}). Setting $t$ to be the concatenation of $s'[1..\journal{j-1}]$ and $s'[\journal{j+k}\ldots m]$, it is clear that $s'$ is in $S$. 
Thus $S$ contains all strings of length $m$ having at least one repeated $k$-mer. 
Since there are \journal{${ {m-k+1} \choose 2 }$ choices for $(i,j)$, i.e. less than ${ m \choose 2 }$,} and $4^{m-k}$ choices for $t$, the cardinality of $S$ is at most ${m \choose 2 } 4^{m-k}$, which yields the result.
\end{proof}

\arxiv{

\begin{figure}[t]
\centering
\includegraphics[scale=0.3]{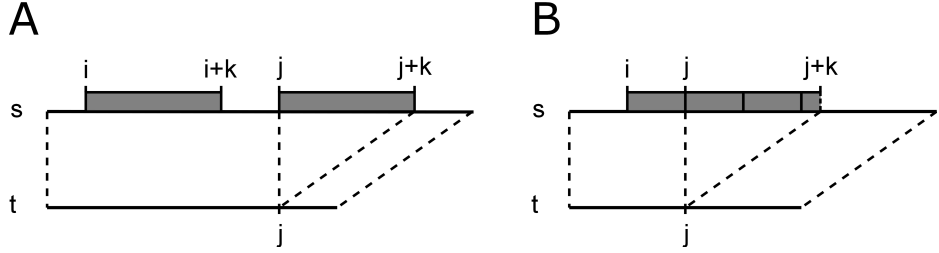}
\caption{
\journal{
Illustration of the proof of Lemma~\ref{lem:nb_strings_repeated_kmer}, showing the definition of the set $S$ of strings of length $m$ containing identical $k$-mers at positions $1 \leq i < j \leq (m-k+1)$. Panel A illustrates the case when $(j-i)\geq k$: the $k$-mers $s[i\ldots i+k-1]$ and $s[j\ldots j+k-1]$ do not overlap. When $(j-i) < k$ (Panel B), the $k$-mer $s[i\ldots i+k-1]$ (also equal to $s[j\ldots j+k-1]$) consists of repetitions of the string $s[i\ldots j-1]$.}
\label{fig:lblin}}
\end{figure}
}

\begin{restatable}{lem}{nbstringsnonrepeatedkmer}
\label{lem:nb_strings_non_repeated_kmer}
Given $0<\epsilon<1$, let $n=\lceil 3 \epsilon^{-1} \rceil$ and $k = \lceil 1+(2+\epsilon)\log_4(2n)\rceil$. The number of DNA strings of length $(n+k-1)$ which start with the same $k$-mer, and do not contain any repeated $(k-1)$-mer, is strictly greater than $4^{n(1-\epsilon)}$.
\end{restatable}
\begin{proof}
Note that $k<n$, thus $k > (1+(2+\epsilon)\log_4(n+k-1))$ and $4^{-k+1} < (n+k-1)^{(-2-\epsilon)}$.
 Using Lemma~\ref{lem:nb_strings_repeated_kmer}, there are at least $(4^{n+k-1}(1-{n+k-1 \choose 2}4^{-k+1})) > (4^{n+k-1}(1-\frac{1}{2(n+k-1)^\epsilon}))$ strings of length $(n+k-1)$ where each $(k-1)$-mer is unique. Thus, each string has exactly $n$ $k$-mers that are all distinct. By binning these strings with respect to their first $k$-mer, there exists a $k$-mer $k_0$ such that there are at least $4^{n-1}(1-\frac{1}{2(n+k-1)^\epsilon})$ strings starting with $k_0$, which do not contain any repeated $(k-1)$-mer. The following inequalities hold: $4^{-1} > 4^{-n\epsilon/2}$ and $(1-\frac{1}{2(n+k-1)^\epsilon}) > \frac{1}{2} > 4^{-n\epsilon/2}$. Thus, $4^{n-1}(1-\frac{1}{2(n+k-1)^\epsilon}) > 4^{n(1-\epsilon)}$.

\end{proof}

\begin{lem}
    \label{lem:different_strings}
	Two different strings of length $(n+k-1)$ starting with the same $k$-mer and not containing any repeated $(k-1)$-mer correspond to two different linear de Bruijn graphs.
\end{lem}
\begin{proof}
For two different strings $s_1$ and $s_2$ of length $(n+k-1)$, which start with the same $k$-mer and do not contain any repeated $(k-1)$-mer, observe that their sets of $k$-mers cannot be identical. Suppose they were, and consider the smallest integer $i$ such that $s_1[i\ldots i+k-2] = s_2[i\ldots i+k-2]$ and $s_1[i\journal{+k-1}]\neq s_2[i\journal{+k-1}]$. The $k$-mer $s_1[i\ldots i+k-1]$ appears in $s_2$, at some position $j\neq i$. Then $s_2[i\ldots i+k-2]$ and $s_2[j\ldots j+k-2]$ are identical $(k-1)$-mers in $s_2$, which is a contradiction. Thus, $s_1$ and $s_2$ correspond to different sets of $k$-mers, and therefore correspond to two different linear de Bruijn graphs.
\end{proof}

\lineargraphslb*
\begin{proof}
Assume for the sake of contradiction that for all linear de Bruijn graphs, the output of \const requires less than $2(1-\epsilon)$ bits per $k$-mer. 
Thus for a fixed $k$-mer length, the number of outputs $\const(S)$ for sets of $k$-mers $S$ of size $n$ is no more than $2^{2n(1-\epsilon)}$. 
Lemma~\ref{lem:nb_strings_non_repeated_kmer} provides values $(k,n,k_0)$, for which there are more strings starting with a $k$-mer $k_0$ and containing exactly $n$ $k$-mers with no duplicate $(k-1)$-mers (strictly more than $2^{2n(1-\epsilon)}$) than the number of outputs $\const(S)$ for $n$ $k$-mers. 

By the pigeonhole principle, there exists a navigational data structure constructor $\const(S)$ 
that takes the same values on two different strings $s_1$ and $s_2$ that start with the same $k$-mer $k_0$ and do not contain repeated $(k-1)$-mer. 
By Lemma~\ref{lem:different_strings}, $\const(S)$ takes the same values on two different sets of $k$-mers $S_1$ and $S_2$ of cardinality $n$.
Let $p$ be the length of longest prefix common to both strings. 
Let $k_1$ be the $k$-mer at position $(p-k+1)$ in $s_1$. 
Note that $k_1$ is also the $k$-mer that starts at position $(p-k+1)$ in $s_2$. 
By construction of $s_1$ and $s_2$, $k_1$ does not appear anywhere else in $s_1$ or $s_2$. 
Moreover, the $k$-mer at position $(p-k)$ in $s_1$ is different to the $k$-mer at position $(p-k)$ in $s_2$. 
In a linear de Bruijn graph corresponding to a string where no $(k-1)$-mer is repeated, each node has at most one out-neighbor. 
Thus, the out-neighbor of $k_1$ in the de Bruijn graph of $S_1$ is different to the out-neighbor of $k_1$ in the de Bruijn graph of $S_2$, i.e. 
$\neighb(\const(S_1),k_1) \neq \neighb(\const(S_2),k_1)$, which is a contradiction. 
\end{proof}

\section{Data structure for representing a de Bruijn graph in small space (\dbgfm) \label{sec:dbgfm}}
Recall that a simple path is a path where all the internal nodes have in- and out-degree of 1.
A \emph{maximal} simple path is one that cannot be extended in either direction.
\revision{It can be shown that there exists a unique set of edge-disjoint maximal simple paths that completely covers the dBG,}
and each path $p$ with $|p|$ nodes can be represented compactly 
as a string of length $k + |p| - 1$.
We can thus represent a \dbg $S$ containing $n$ \kmers 
as a set of strings corresponding to 
the 
maximal simple paths,
denoted by $sp_k(S)$.
Let $c_k(S)=|sp_k(S)|$ be the number of maximal simple paths, and
let $s$ to be the concatenation of all members of $sp_k(S)$ in arbitrary order, 
separating each element by a symbol not in $\Sigma$ (e.g. \$).
Using the same idea as in Section~\ref{sec:lblin}, 
we can represent a \dbg using $s$ in 
\journal{
\begin{align*}
2|s| &= \sum_{p \in sp_k(S)} 2(|p| + k)
\\
&\leq 2(n + (k + 2)c_k(S))\text{ bits.}\end{align*}}
However, this representation requires an inefficient linear scan in order to answer a membership query.

We propose the use of the FM-index of $s$ to speed up query time at the cost of more space.
The FM-index~\citep{fmindex} is a full-text index which is based on the Burrows-Wheeler transform~\citep{bwt,bwttextbook} developed for text compression.
It has previously been used to map reads to a reference genome~\citep{bwa,bowtie2,soap2},
perform \emph{de novo} assembly~\citep{simpson10,simpson12,fermi},
and represent the \dbg for the purpose of exploring
genome characteristics prior to \emph{de novo} assembly~\citep{simpson13}.

The implementation of the FM-index  
stores the Huffman-coded Burrows-Wheeler transform of $s$
along with two associated arrays and some $o(1)$ space overhead.
Our software, called \dbgfm\footnote{\revision{Source code available at \url{http://github.com/jts/dbgfm}}}, follows the implementation of~\citet{fmindex}, and
we refer the reader there for a more thorough description of how the FM-index works.
Here, we will only state its most relevant properties.
It allows us to count the number of occurrences of an arbitrary pattern $q$ in $s$ in $O(|q|)$ time.
In the context of \dbgs, we can test for the membership of a $k$-mer in $S$ in time $O(k)$.
Two sampling parameters ($r_1$ and $r_2$)
trade-off the size of the associated arrays with the query time.
For constructing the FM-index, there are external memory algorithms that do not use
more intermediate memory than the size of the final output~\citep{fmindex}.
The space usage of \dbgfm is
\journal{\begin{align*}
&\phantom{{}\leq{}}
|s|(H_0(s) + 96r_1^{-1} + 384r_2^{-1}) + o(1) 
\\
&\leq
n(H_0(s) + 96r_1^{-1} + 384r_2^{-1})
(1 + \frac{k + 2}{n}c_k(S)) + o(1)\text{ bits,}
\end{align*}
where $H_0$ is the zeroth order entropy~\citep{bwttextbook}:
$H_0(s)  = - \sum_{c\in \Sigma \cup \{\$\}} f_c \lg f_c$, and
$f_c$ is the frequency of character $c$ in $s$.}
Note that for our five character alphabet
$H_0$ is at most $\lg 5$.

As the value of $c_k(S)$ approaches one, $f_{\$}$ approaches $0$ and hence
the upper bound on $H_0$ approaches $2$.
If we further set the sampling parameters  to be inversely proportional to $n$,
the space utilization approaches at most $2n$ bits.
However, this results in impractical query time and, more realistically, 
typical values for the  sampling parameters are $r_1 = 256$ and $r_2 = 16384$,
resulting in at most $2.32n$ bits as $c_k(S)$ approaches 1.
For the error-free human genome with $k=55$, there are $c_{55}(S) = 12.9 \cdot 10^6$ maximal simple paths
and $n = 2.7 \cdot 10^9$ \kmers.
The resulting $H_0(S)$ is at most $2.03$, and
the space utilization is at most 2.43 bits per \kmer.



An additional benefit of the FM-index is that it allows constant-time access to the in-neighbors of nodes ---
every edge is part of a simple path, so we can query the FM-index for the symbols preceding a \kmer $x$.
Thus, \dbgfm is a membership data structure but supports faster in-neighbor queries.
However we note that this is not always the case when reverse complementarity is taken into account.

We wanted to demonstrate how the \dbgfm data structure could be incorporated into an existing assembler. 
We chose ABySS, a popular \emph{de novo} sequence assembly tool used in large-scale genomic projects~\citep{abyss}.
In modifying ABySS to look up \kmer membership using \dbgfm, 
we replace its hash table with a simple array.
\dbgfm associates each \kmer with an integer called a suffix array index~(\sai),
which could be used to index the simple array.
However,
some of the \kmers of the \dbgfm string include a simple path separator symbol, \$,
and, hence, not every \sai corresponds to a node in the \dbg.
We therefore use 
a rank/select data structure \citep{gonzalez2005practical} to translate 
the \sais into a contiguous enumeration of the nodes,
which we then use to index our simple array.
We also modified the graph traversal strategy in order to maximize the number
of in-neighborhood queries, which are more efficient than out-neighborhood or 
membership queries.

\section{Algorithm to enumerate the maximal simple paths of a de Bruijn graph in low memory (\bcalm)\label{sec:bcalm}}
The \dbgfm data structure of Section~\ref{sec:dbgfm} can construct and represent a 
\dbg in low space from the set of maximal simple paths ($sp_k(S)$).
However, constructing the paths (called compaction) generally requires loading the \kmers into memory, 
which would require large intermediate memory usage.
Because our goal is a data structure that is low-memory during both construction and the final
output,
we developed an algorithm for compacting de Bruijn graphs in low-memory (\bcalm) \footnote{\revision{Source code available at \url{http://github.com/Malfoy/bcalm}}}.

\subsection{Algorithm description}
Our algorithm is based on the idea of minimizers, first introduced by~\citet{roberts04,roberts04b} 
and later used by~\citet{yanglipaper}.
The $l$-minimizer of a string $u$ is the smallest $l$-mer that is a sub-string of $u$ 
(we assume there is a total ordering of the strings, e.g. lexicographical).
We define $\lmin(u)$ (respectively, $\rmin(u)$) to be the $\ell$-minimizer of the 
$(k-1)$-prefix (respectively suffix) of $u$.
We refer to these as the left and right minimizers of $u$, respectively.
We use minimizers because of the following observation:
\begin{restatable}{obs}{rminequalslmin}
\label{obs:rminequalslmin}
For two strings $u$ and $v$, if $u\rightarrow v$, then $\rmin(u) = \lmin(v)$.
\end{restatable}

We will first need some definitions.
Given a set of strings $S$, we say that $(u, v) \in S^2$ are \emph{compactable} in a set $V \subseteq S$ if 
$u\rightarrow v$ and, 
$\forall w \in V$, if $w\rightarrow v$ then $w=u$ and if $u\rightarrow w$ then $w=v$. 
The compaction operation is defined on a pair of compactable strings $u,v$ in  $S$. 
It replaces $u$ and $v$ by a single string
$w=u \cdot v[k+1\ldots |v|]$ where '$\cdot$' is the string concatenation operator.
Two strings $(u,v)$ are \emph{$m$-compactable} in $V$ if they are compactable in $V$ and 
if $m=\rmin(u) = \lmin(v)$.
The $m$-compaction of a set $V$ is obtained from $V$ by applying the compaction operation as much as possible in any order to all pairs of strings that are $m$-compactable in $V$. 
It is easy to show that the order in which strings are compacted does not lead to different $m$-compactions.
Compaction is a useful notion because a simple way to obtain the simple paths
is to greedily perform compaction as long as possible.
In the following analysis, we identify a string $u$ with the path 
$p = u[1\ldots k] \rightarrow u[2\ldots (k+1)] \rightarrow \ldots \rightarrow u[(|u|-k+1)\ldots |u|]$ of all the $k$-mers of $u$ in consecutive order.

\arxiv{
\begin{algorithm}[p!]
\begin{algorithmic}[1]
\STATE \textbf{Input:} Set of k-mers $S$, minimizer size $\ell<k$
\STATE \textbf{Output:} Sequences of all simple paths in the de Bruijn graph of $S$
 \STATE Perform a linear scan of $S$ to get the frequency of all $l$-mers (in memory)\label{algstep:lmercounting}
 \STATE Define the ordering of the minimizers, given by their frequency in $S$\label{algstep:minimizersorder}
 \STATE Partition $S$ into files $F_m$ based on the minimizer $m$ of each $k$-mer\label{algstep:partition}
 \FOR{each file $F_m$ in increasing order of $m$\;\label{algstep:for-each-file}}
  \STATE $C_m \leftarrow$ $m$-compaction of $F_m$ (performed in memory)\;\label{algstep:compact-graph}
  \FOR{each string $u$ of $C_m$}
    \STATE $B_{min} \leftarrow $ min(\lmin($u$),\rmin($u$))\;\label{algstep:bmin}
    \STATE $B_{max} \leftarrow $ max(\lmin($u$),\rmin($u$))\;
    \IF{$B_{min} \le m \text{ and } B_{max} \le m$}
        \STATE Output $u$\;\label{algstep:output}
    \ELSIF{$B_{min} \le m \text{ and } B_{max} > m$}
        \STATE Write $u$ to $F_{B_{max}}$\;
    \ELSIF{$B_{min} > m \text{ and } B_{max} > m$}
        \STATE Write $u$ to $F_{B_{min}}$\;\label{algstep:movemin}
    \ENDIF
   \ENDFOR
   \STATE Delete $F_m$
  \ENDFOR
  \end{algorithmic}
  \caption{\bcalm: Enumeration of all maximal simple paths in the \dbg\label{alg:bcalm}}
\end{algorithm}
}

We now give a high-level overview of Algorithm~\ref{alg:bcalm}. 
The input is a set of \kmers $S$ and a parameter $\ell < k$ which is the minimizer size.
For each $m \in \Sigma^\ell$, we maintain a file $F_m$ in external memory.
Each file contains a set of strings, and we will later prove that at any point during the execution, 
each string is a sub-path of a simple path (Lemma~\ref{lem:correctpath}).
Moreover, we will show that at any point of execution, the multi-set of \kmers appearing in the strings and in the output does not change and is always $S$ (Property~\ref{prop:kmersdontchange}).

At line~\ref{algstep:partition}, we partition the \kmers into the files, according to their $\ell$-minimizers.
Next, each of the files is processed, starting from the file of the smallest minimizer in increasing order
(line~\ref{algstep:for-each-file}).
For each file, we load the strings into memory and $m$-compact them
(line~\ref{algstep:compact-graph}),
with the idea being that the size of each of the files is kept small enough so that memory usage is low.
The result of the compaction is a new set of strings, each of which is then either written to one of the files 
that has not been yet processed or output as a simple path.
\journal{A step-by-step execution of Algorithm~\ref{alg:bcalm} on a toy example is given in Figure~\ref{fig:bcalm_example}.}

\arxiv{

\begin{figure}
\centering
\includegraphics[width=.8\textwidth]{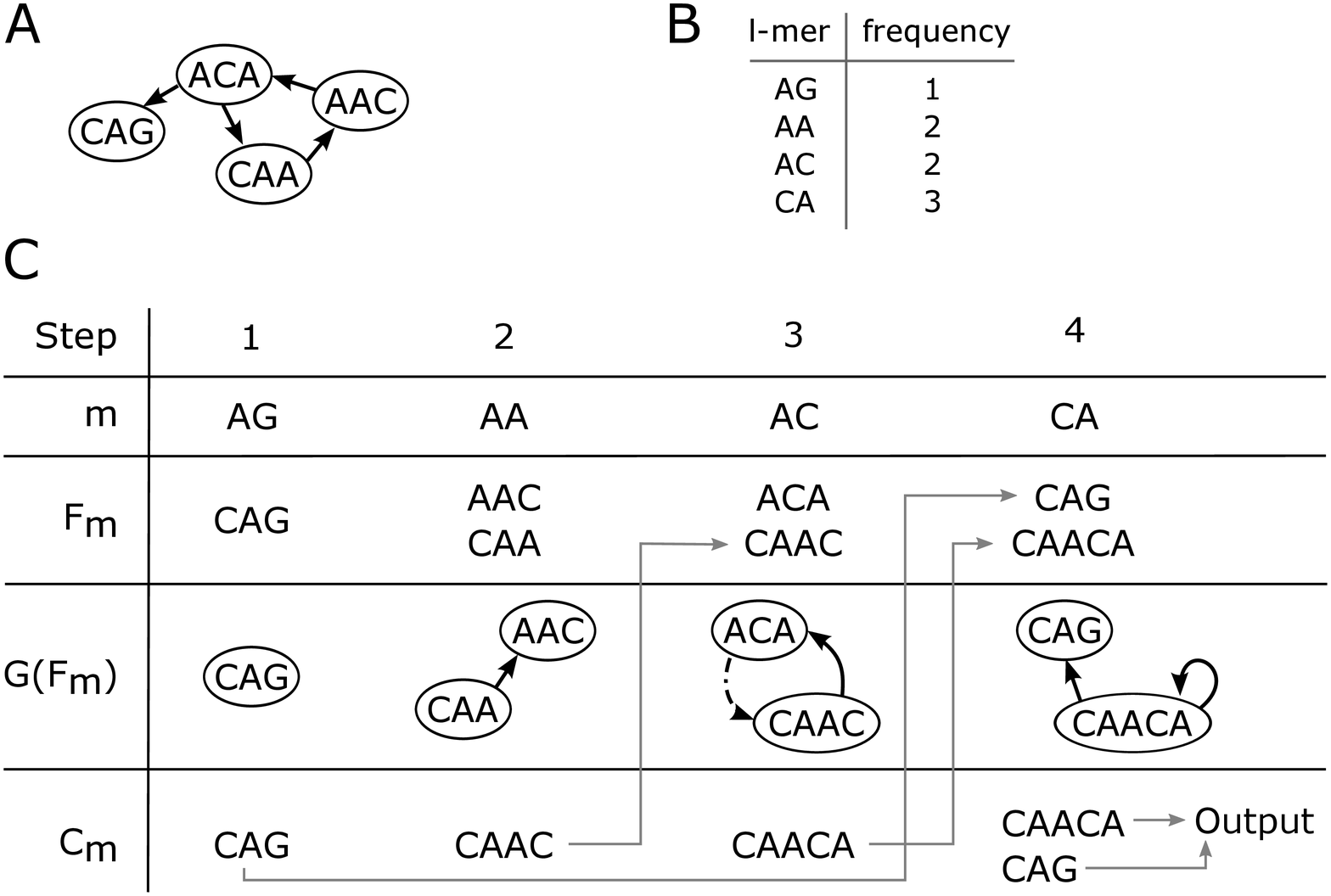}
\caption{
\journal{
An execution of Algorithm~\ref{alg:bcalm} on a toy example. (Panel A) The set $S=\{\text{ACA}, \text{AAC}, \text{CAA}, \text{CAG}\}$ of $k$-mers is represented as a de Bruijn graph ($k=3$). (Panel B) The frequencies of each $l$-mer present in the $k$-mers is given in increasing order ($l = 2$). (Panel C) Steps of the algorithm; initially, each set $F_m$ contains all $k$-mers having minimizer $m$. For each minimizer $m$ (in order of frequency), perform the $m$-compaction of $F_m$ and store the result in $C_m$; the grey arrows indicate where each element of $C_m$ is written to, either the file of a higher minimizer or the output. The row $G(F_m)$ shows a graph where solid arrows indicate $m$-compactable pairs of strings in $F_m$. The dash-dot arrow in Step 3 indicates that the two strings are compactable in $F_m$ but not $m$-compactable; in fact, they are not compactable in $S$.}
\label{fig:bcalm_example}}
\end{figure}
}

The rule of choosing which file to write to is based on the left and right minimizers of the string.
If both minimizers are no more than $m$, then the string is output as a simple path
(line~\ref{algstep:output}).
Otherwise, we identify $m'$, the smallest of the two minimizers that is bigger than $m$,
and write the string to the file $F_{m'}$.
Finally, the file $F_m$ is discarded, and the next file is processed.
We will show that the rule for placing the strings into the files 
ensures that as each file $F_m$ is processed (line~\ref{algstep:for-each-file}),
it will contain every \kmer that has $m$ as a minimizer (Lemma~\ref{lem:everykmerinfm}).
We can then use this to prove the correctness of the algorithm (Theorem~\ref{thm:bcalmthm}).

There are several implementation details that make the algorithm practical.
First, reverse complements are supported in the natural way by identifying 
each \kmer with its reverse complement and letting the minimizer be the smallest $l$-mer
in both of them.
Second, we avoid creating $4^l$ files, which may be more than the file system supports.
Instead, we use virtual files and group them together into a smaller number of 
physical files. This allowed us to use $l=10$ in our experiments.
Third, when we load a file from disk (line~\ref{algstep:compact-graph}) 
we only load the first and last \kmer of each string, since the middle part is never used by the 
compaction algorithm. 
We store the middle part in an auxiliary file and use a pointer to keep track of it within the strings 
in the $F_m$ files.

Consequently, the algorithm memory usage depends on the number of strings in each file $F_m$, but not on the total size of those files. For a fixed input $S$, the number of strings in a file $F_m$ depends on the minimizer length $l$ and the ordering of minimizers. When $l$ increases, the number of $(k-1)$-mers in $S$ that have the same minimizer decreases. Thus, increasing $l$ yields less strings per file, which decreases the memory usage. We also realized that, when highly-repeated $l$-mers are less likely to be chosen as minimizers, the sequences are more evenly distributed among files. 
We therefore perform in-memory \lmer counting
(line~\ref{algstep:lmercounting}) to obtain a sorted frequency table of all \lmers. 
This step requires an array of $64|\Sigma|^l$ bits to store the count of each $l$-mer in $64$ bits,
which is negligible memory overhead for small values of $\ell$ ($8$ MB for $l=10$). 
Each $l$-mer is then mapped to its rank in the frequency array, to create a total ordering of minimizers (line~\ref{algstep:minimizersorder}).
Our experiments showed a drastic improvement over lexicographic ordering (results in Section~\ref{sec:results}).

\subsection{Proof of correctness}
In this subsection we give a formal proof of the correctness of BCALM.
\begin{restatable}{prop}{kmersdontchange}
\label{prop:kmersdontchange}
At any point of execution after line~\ref{algstep:partition}, 
the multi set of \kmers present in the files and in the output is $S$.
\end{restatable}
\begin{proof}
We prove by induction.
It is trivially true after the partition step.
In general, note that the compaction operation preserves the multi set of \kmers.
Because the only way the strings are ever changed is through compaction, 
the property follows.
\end{proof}

\begin{restatable}{lem}{everykmerinfm}
For each minimizer $m$, for each $k$-mer $u$ in $S$ such that $\lmin(u)=m$ (resp. $\rmin(u)=m$), $u$ is the left-most (resp. right-most) $k$-mer of a string in $F_m$ at the time $F_m$ is processed.
\label{lem:everykmerinfm}
\end{restatable}
\begin{proof}
We prove this by induction on $m$. Let $m_0$ be the smallest minimizer. All $k$-mers that have $m_0$ as a left or right minimizer are strings in $F_{m_0}$, thus the base case is true. Let $m$ be a minimizer and $u$ be a $k$-mer such that $\lmin(u)=m$ or $\rmin(u)=m$, and assume that the induction hypothesis holds for all smaller minimizers. If $\min(\lmin(u),\rmin(u))=m$, then $u$ is a string in $F_m$ after execution of line~\ref{algstep:partition}. Else, without loss of generality, assume that $m = \rmin(u) > \lmin(u)$. Then, after line~\ref{algstep:partition}, $u$ is a string in $F_{\lmin(u)}$. Let $F_{m^1},\ldots,F_{m^t}$ be all the files, in increasing order of the minimizers, which have a simple path containing $u$ before the maximal-length simple path containing $u$ is outputted by the algorithm. Let $m^i$ be the largest of these minimizers strictly smaller than $m$. By the induction hypothesis and Property~\ref{prop:kmersdontchange}, $u$ is at the right extremity of a unique string $s_u$ in $F_{m^i}$. After the $m^i$-compactions, since $m  = \rmin(s_u) > m^i$, $s_u$ does not go into the output. It is thus written to the next larger minimizer. Since $m = \rmin(u) \leq m^{i+1}$, then it must be that $m^{i+1} = m$, and $s_u$ is written to $F_m$, which completes the induction.
\end{proof}

\begin{restatable}{lem}{correctpath}
\label{lem:correctpath}
In Algorithm~\ref{alg:bcalm}, at any point during execution, each string in $F_m$ corresponds to a sub-path of a maximal-length simple path.
\end{restatable}
\begin{proof}
We say that a string is \emph{correct} if it corresponds to a sub-path of a maximal-length simple path. We prove the following invariant inductively: at the beginning of the loop at line~\ref{algstep:for-each-file}, all the files $F_m$ contain correct strings. The base case is trivially true as all files contain only $k$-mers in the beginning.
Assume that the invariant holds before processing $F_m$. 
It suffices to show that no wrong compactions are made; i.e. if two strings from $F_m$ are $m$-compactable, then they are also compactable in $S$. The contrapositive is proven. Assume, for the sake of obtaining a contradiction, that two strings $(u,v)$ are not compactable in $S$, yet are $m$-compactable in $F_m$ at the time it is processed.
Without loss of generality, assume that there exists $w\in S$ such that $u \rightarrow w$  and $w \ne v$. 
Since $u\rightarrow v$ and $u\rightarrow w$, $m = \rmin(u) = \lmin(v) = \lmin(w)$. Hence, by Lemma~\ref{lem:everykmerinfm}, $w$ is the left-most $k$-mer of a string in $F_m$ at the time $F_m$ is processed. This contradicts that $(u,v)$ are $m$-compactable in $F_m$ at the time it is processed. Thus, all compactions of strings in $F_m$ yield correct strings, and the invariant remains true after $F_m$ is processed. 
\end{proof}


\begin{restatable}{thm}{bcalmthm}
\label{thm:bcalmthm}
The output of Algorithm~\ref{alg:bcalm} is the set of maximal  simple paths of the de Bruijn graph of $S$.
\end{restatable}
\begin{proof}
By contradiction, assume that there exists a maximal-length simple path $p$ that is not returned by Algorithm~\ref{alg:bcalm}. 
Every input $k$-mer is returned by Algorithm~\ref{alg:bcalm} in some string, and by Lemma~\ref{lem:correctpath}, every returned string corresponds to a sub-path of a maximal-length simple path. Then, without loss of generality, assume that a simple path of $p$ is split into sub-paths $p_1,p_2,\ldots,p_i$ in the output of Algorithm~\ref{alg:bcalm}.  Let $u$ be the last $k$-mer of $p_1$ and $v$ be the first $k$-mer of $p_2$. Let $m=\rmin(u)= \lmin(v)$ (with Observation~\ref{obs:rminequalslmin}). By Lemma~\ref{lem:everykmerinfm}, $u$ and $v$ are both present in $F_m$ when it is processed. As $u$ and $v$ are compactable in $S$ (to form $p$), they are also compactable in $F_m$ and thus the strings that include $u$ and $v$ in $F_m$ are compacted at line~\ref{algstep:compact-graph}. This indicates that $u$ and $v$ cannot be extremities of $p_1$ and $p_2$, which yields a contradiction.
\end{proof}

\section{Results\label{sec:results}}
We tested the effectiveness of our algorithms on \journal{the de Bruijn graphs} of two sequencing datasets. 
Experiments in Tables~\ref{tbl:pipelineanalysis},~\ref{tbl:externalcomparison} and~\ref{tbl:ordering} were run on a single core of a desktop computer equipped with an Intel i7 3.2 GHz processor, 8 GB of memory and a 7200 RPM hard disk drive. \revision{Experiments in Tables~\ref{tbl:mdependence} and~\ref{tbl:abyssdbgfm} were run on a single core of a cluster node with 24 GB of memory and 2.67 GHz cores}.
In all experiments, at most 300 GB of temporary disk space was used. 
The first dataset is 36 million 155bp Illumina human chromosome 14 reads (2.9 GB compressed fastq) from the GAGE benchmark~\citep{gage}. The second dataset is 1.4 billion Illumina 100bp reads (54 GB compressed fastq) from the NA18507 human genome (SRX016231).
We first processed the reads with \kmer counting software, which is the first step of most assembly pipelines.
\revision{We used a value of $k = 55$ as we found it gives reasonably good results on both datasets.}
We used \DSK~\citep{dsk}, a software that is designed specifically for low memory usage 
and can also filter out low-count \kmers as they are likely due to sequencing errors
(we used $<5$ for chr14 and $<3$ for the whole genome). 

\arxiv{
\begin{table}
\begin{minipage}{\linewidth}
\setlength{\tabcolsep}{12pt}
\centering
\small
\begin{tabulary}{\linewidth}{@{}LLLL@{}}
\toprule
Dataset & \DSK &
\bcalm & 
\dbgfm 
 \\
\midrule
\multirow{2}{*}{Chromosome 14} & 
43 MB & 
19 MB & 
38 MB 
 \\
 & 
25 mins & 
15 mins &
7 mins  \\

\multirow{2}{*}{Whole human genome}& 
1.1 GB & 
43 MB &  
1.5 GB   
\\
 & 
5 h & 
12 h & 
7 h  \tabularnewline 

\bottomrule \\
\end{tabulary}
\end{minipage}
\caption{
Running times (wall-clock) and memory usage of \DSK, \bcalm and \dbgfm  construction on the human chromosome 14 and whole human genome datasets ($k=55$ and $\ell=10$ for both).
\label{tbl:pipelineanalysis}}
\end{table}
}

First, we ran \bcalm on the of \kmers computed by \DSK.
The output of \bcalm 
was then passed as input to \dbgfm, which constructed the FM-index.
Table~\ref{tbl:pipelineanalysis} shows the resulting time and memory usage of \DSK, \bcalm, and \dbgfm.
For the whole genome dataset, \bcalm used only 43 MB of memory to 
take a set of $2.5\cdot10^9$ 55-mers and 
output 40 million sequences of total length 4.6 Gbp. 
\dbgfm represented these paths in an index of size 1.5 GB, 
using no more than that memory during construction.
The overall construction time, including \DSK, was roughly 24 hours. \revision{In comparison, a subset of this dataset was used to construct the data structure of~\citet{cascading} in 30.7 hours.}

\arxiv{
\begin{table}
\begin{minipage}{\linewidth}
\setlength{\tabcolsep}{12pt}
\centering
\small
\begin{tabulary}{\linewidth}{@{}LLLL@{}}
\toprule
 & DBGFM & Salikhov~\emph{et al.} 
 &  Conway \& Bromage \\
\midrule
chr14 & 
38.0 MB  &  
94.9 MB &   
$>875$ MB \\

Full human dataset &  
1462 MB  &  
2702 MB  & 
$>22951$ MB    
\\
\bottomrule \\
\end{tabulary}
\end{minipage}
\caption{
Memory usage of de Bruijn graph data structures, on the human chromosome 14 and whole human genome datasets ($k=55$ for both).
We did not run the algorithm of Conway and Bromage because our machine does not have sufficient memory
for the whole genome.
Instead, we report the theoretical size of their data structure, assuming that it would be constructed from the output of \DSK.
As described in~\cite{conway_succinct_2011},
this gives a lower bound on how well their implementation could perform.
\label{tbl:externalcomparison}}
\end{table}
}

We compared the space utilization of our \dbgfm representation with 
that of other low space data structures, 
\citet{cascading} and ~\citet{conway_succinct_2011}
(Table~\ref{tbl:externalcomparison}).
Another promising approach is that of ~\citet{sadakane},
but they do not have an implementation available.
We use $3.53$ bits per $k$-mer (38.0 MB total) for chr14 and
4.76 bits per \kmer (1462 MB total) for the whole-genome.
This is a $60\%$ and $46\%$ improvement over the state-of-the art, respectively. 

\arxiv{
\begin{table}
\begin{minipage}{\linewidth}
\setlength{\tabcolsep}{12pt}
\centering
\small
\begin{tabulary}{\linewidth}{@{}LLLL@{}}
\toprule
 Ordering type  
 & Lexicographical
 & Uniformly Random 
 &   $l$-mer frequency  \\
\midrule 
Memory usage  & 804 MB  &  222 MB  & 19 MB   
\\
\bottomrule \\
\end{tabulary}
\end{minipage}
\caption{
Memory usage 
of \bcalm with three different minimizer orderings: lexicographical, uniformly random, 
and according to increasing \lmer frequencies.
The dataset used is the human chromosome 14 with $k=55$ and $l=8$.
\label{tbl:ordering}}
\end{table}
}

During algorithm development, 
we experimented with different ways to order the minimizers and the effect on memory usage
(Table~\ref{tbl:ordering}).
Initially, we used the lexicographical ordering, but experiments with the chromosome 14 dataset showed it was 
a poor choice, resulting in 804 MB memory usage with $l=8$.
The lexicographically smallest $l$-mer is $m_0=A^l$,
which is overly abundant in human chromosomes for $l\leq 10$, resulting in a large file $F_{m_0}$. 
In a second attempt, 
we created a uniformly random ordering of all the \lmers.
While $A^l$ is no longer likely to have a small value,
it is still likely that there is a highly repeated \lmer that comes early in the ordering,
resulting in 222 MB memory usage.
Finally, we ordered the \lmers according to their frequency in the dataset.
This gave a memory usage of 19 MB, resulting in 
a 40-fold improvement 
over the initial lexicographical ordering. The running times of all three orderings were comparable.
We also evaluated the effect that the minimizer size $\ell$ has on memory usage \revision{and running time} 
(Table~\ref{tbl:mdependence}). 
Larger $\ell$ will generally lead to smaller memory usage, however
we did not see much improvement past $l=8$ on this dataset.

\arxiv{
\begin{table}
\begin{minipage}{\linewidth}
\setlength{\tabcolsep}{12pt}
\centering
\small
\begin{tabulary}{\linewidth}{@{}LLLLLLLL@{}}
\toprule
 Minimizer size $l$  
 & \ 2
 & 4
 &   6 
 & 8
 &  10 
\\
\midrule
Memory usage  &
9879 MB &   
3413 MB  &   
248 MB  &   
19 MB  &  
19 MB    
\\
\revision{Running time} & 
\revision{27m19s} &  
\revision{22m2s} & 
\revision{20m5s} & 
\revision{18m39s} & 
\revision{21m4s} 
\\
\bottomrule \\
\end{tabulary}
\end{minipage}
\caption{
Memory usage \revision{and wall-clock running time} of \bcalm with increasing values of minimizer sizes $l$ on the chr14 data.
By grouping files into virtual files, these values of $l$ require respectively 4, 16, 64, 256 and 1024 physical files on disk. The ordering of minimizers used is the one based on $l$-mer counts.
\label{tbl:mdependence}}
\end{table}}

\arxiv{
\begin{table}
\begin{minipage}{\linewidth}
\setlength{\tymax}{0.5\linewidth}
\centering
\small
\begin{tabulary}{\textwidth}{@{}LLLLLLL@{}}
\toprule
Data structure & Memory usage & Bytes\slash \kmer & dBG (B\slash \kmer) & Data (B\slash \kmer) & Overhead (B\slash \kmer) & Run time\\
\midrule
sparsehash & 2429 MB & 29.50 & 16   & 8 & 5.50 & 
\revision{14m4s} \\ 
\dbgfm      & 739 MB  & 8.98  & 0.53 & 8 & 0.44 & 
\revision{21m1s} \\ 
\bottomrule \\
\end{tabulary}
\end{minipage}
\caption{Memory usage and run time (wall clock) of the ABySS hash table implementation (sparsehash) and of the \dbgfm implementation, 
using a single thread to process the human chromosome 14 data
set. 
The dBG bytes/$k$-mer column corresponds to the space taken by encoded $k$-mers for  sparsehash, and the FM-index for \dbgfm. The Data bytes/$k$-mer column corresponds to associated data. The Overhead bytes/$k$-mer corresponds to the hash table and heap overheads, as well as the rank/select bit array.
The run time of the \dbgfm row does not
include the time to construct the \dbgfm representation.}
\label{tbl:abyssdbgfm}
\end{table}
}

Finally, we evaluated the performance of ABySS using \dbgfm compared with that of the 
hash table implementation (Table~\ref{tbl:abyssdbgfm}). \revision{Note, however, that only the graph traversal and marking steps were implemented in the \dbgfm version, and none of the graph simplifications steps of ABySS.}
The \dbgfm version used $70\%$ less memory, albeit 
the hash version was \revision{$33\%$} faster, 
indicating the time/space trade-off inherent in the FM-index. 
In addition to storing the graph,
ABySS associates data with each \kmer: 
the count of each \kmer and its reverse complement (two 16 bits counters),
the presence or absence of the four possible in- and out-edges (8 bits), 
three boolean flags indicating whether the \kmer and its reverse complement
have been visited in graph traversal (2 bits), and 
whether they have been removed (1 bit).
While in the hash implementation, the graph structure takes \revision{54}\% of the memory,
in the \dbgfm version it only used \revision{6}\% of memory.
This indicates that further memory improvements can be made by optimizing 
the memory usage of the associated data.


\section{Conclusion}
This paper has focused on pushing the boundaries 
of memory efficiency of de Bruijn graphs. 
Because of the speed/memory trade-offs involved, 
this has come at the cost of slower data structure construction and query times.
Our next focus will be on improving these runtimes through optimization and parallelization of our algorithms.

We see several benefits of low-memory de Bruijn graph data structures in genome assembly.
First, there are genomes like the 20 Gbp white spruce which are an order of 
magnitude longer than the human which cannot be assembled by most assemblers,
even on machines with a terabyte of RAM.
Second, even for human sized genomes, the memory burden poses 
unnecessary costs to research biology labs.
Finally, in assemblers such as ABySS that store the \kmers explicitly,
memory constraints can prevent the use of large $k$ values.
With \dbgfm, the memory usage becomes much less dependent on $k$,
and allows the use of larger $k$ values 
to improve the quality of the assembly.

Beyond genome assembly, our work is also relevant to many \emph{de novo} sequencing applications where large de Bruijn graphs are used, e.g. assembly of transcriptomes and meta-genomes~\citep{trinity,raymeta}, and \emph{de novo} genotyping~\citep{cortex}.





\subsubsection*{Acknowledgements}

The authors would like to acknowledge anonymous referees from a previous submission for their helpful suggestions and for pointing us to the paper of \citet{gagie2012}.
The work was partially supported by NSF awards
    DBI-1356529, IIS-1421908, and CCF-1439057 to PM.

\bibliographystyle{jcompbiol}
\bibliography{main}

\makeatletter
\@fpsep\textheight
\makeatother

\end{document}